\documentclass[12pt,a4paper]{article}
\usepackage[T1]{fontenc}
\usepackage[utf8]{inputenc}
\usepackage{lmodern,mathptmx,amsmath,amssymb,amsthm,mathtools,graphicx,booktabs}
\usepackage[a4paper,margin=2cm]{geometry}
\usepackage[numbers,square]{natbib}
\usepackage{hyperref}
\usepackage{microtype}
\hypersetup{pdftitle={The LR1 Lifetime Distribution: Properties, Nonmonotone Hazard, and Statistical Inference},pdfauthor={Lalita Singh and Rohit Patna}}
\numberwithin{equation}{section}
\newtheorem{theorem}{Theorem}
\newtheorem{proposition}[theorem]{Proposition}
\newtheorem{corollary}[theorem]{Corollary}

\newcommand{\E}{\mathbb E}
\newcommand{\Prob}{\mathbb P}

\newcommand{\Exp}{\operatorname{Exp}}
\newcommand{\Gam}{\operatorname{Gamma}}
\newcommand{\dd}{\mathrm d}

\title{The LR1 Lifetime Distribution: Properties, Nonmonotone Hazard, and Statistical Inference}
\author{%
  Lalita Singh\thanks{Corresponding author: mail}\\
  Agra College, Agra
  \and
  Rohit Patawa\\
  Agra College, Agra
}
\date{}
\makeatletter
\begin{document}

\maketitle
\begin{abstract}
The LR1 distribution links an exponential lifetime and a four-stage gamma lifetime through a single positive parameter $\lambda$. Their rates, $2\lambda$ and $\lambda/2$, have a fixed ratio, whereas the exponential mixing probability $\lambda^2/(1+\lambda^2)$ changes with the parameter. This restriction produces a hazard that decreases near zero and eventually increases toward $\lambda/2$ for every $\lambda>0$. We establish this behavior analytically and obtain an explicit parameter threshold separating decreasing densities from densities with an interior local peak. Distributional calculations provide survival probabilities, moments, transforms, truncated moments, quantiles, order statistics, and reliability summaries. The density at zero establishes identifiability, and strict decrease of the mean ensures a unique moment estimate for every positive sample mean. Likelihood inference includes the score, observed information, a constrained expectation--maximization algorithm, and asymptotic confidence limits; spacing estimation and the likelihood under independent right censoring are also described. Eight Monte Carlo settings, comprising 4,000 generated datasets, quantify the error of likelihood and moment estimators and the coverage of log-scale intervals. The results apply in a fixed working time unit and concern mathematical properties and correctly specified simulation, rather than demonstrated superiority on observed data.
\end{abstract}

\section{Introduction}
An exponential waiting time and a sequence of exponential stages represent different mechanisms for a positive duration. The former permits failures immediately after the origin and has a constant hazard; a gamma duration with integer shape describes the accumulated time through several stages. A mixture of these mechanisms can accommodate heterogeneity while retaining simple component distributions. The resulting population failure rate nevertheless requires its own analysis, because the composition of the surviving population changes with age. This paper examines that issue for a constrained exponential--gamma distribution in which the gamma shape is four.

The restriction defining LR1 is particularly concrete. For $\lambda>0$, the exponential rate is $2\lambda$, the gamma rate is $\lambda/2$, and the probability of the exponential component is $\lambda^2/(1+\lambda^2)$. Consequently, the component-rate ratio remains four and the gamma component mean is sixteen times the exponential component mean. Increasing $\lambda$ simultaneously shortens both component lifetimes and changes the population composition in favor of the exponential mechanism. The model is therefore a specified relationship between frequencies and durations, rather than a mixture with three freely adjustable parameters. The notation $\lambda$ denotes the same LR1 parameter previously written as $\theta$; this is a change of symbol, not a reparameterization or a change of distribution.

Such constructions belong to the established literature on finite mixtures. Component allocations may describe latent population groups, or they may simply provide a convenient representation of an observable density. The distinction matters when interpreting fitted components. The review of \citet{mclachlan2019} describes the broader mixture framework and its inferential challenges, including identification and numerical likelihood maximization. For lifetime data, survival probabilities, hazards, and conditional remaining lifetimes address different questions; the framework of \citet{cox1984} supplies the relevant survival-analysis context. A parsimonious formula is useful only insofar as its restrictions and the corresponding reliability behavior are understood.

Exponential--gamma mixtures already appear in familiar lifetime distributions. The Lindley distribution of \citet{lindley1958} combines exponential and shape-two gamma components with a common rate. Its distributional and reliability analysis was developed by \citet{ghitany2008}. Further flexibility can arise through transformations and additional parameters, as illustrated by the power Lindley distribution of \citet{ghitany2013}. These examples motivate analytical study of particular restrictions without implying that the underlying act of mixing components is new.

Other lifetime constructions modify the gamma shape or the allocation probabilities. The xgamma model of \citet{sen2016} combines exponential and shape-three gamma components, while \citet{shanker2016} develops the Sujatha distribution through exponential and gamma ingredients. Separate component rates provide another choice: \citet{sakthivel2023} considers an exponential--gamma mixture with a shape-three gamma component and distinct rate parameters. These alternatives show that component shapes, rate restrictions, and weights must all be specified before comparing the resulting distributions. Properties of one member of this broader class cannot be transferred to another merely because both contain gamma components.

LR1 is also related directly to the LR2 model of \citet{singh2026lr2}. To keep the parameter conventions distinct, write the LR2 parameter as $\vartheta$: its component rates are $\vartheta/2$ and $\vartheta^2$, with exponential probability $\vartheta^2/(1+\vartheta^2)$. Those rates cross at $\vartheta=1/2$, whereas LR1 maintains its fixed ratio at every $\lambda>0$. Both models use shape-four gamma components, quadratic allocation probabilities, and several common analytical methods. The present exposition retains that relationship explicitly and concentrates on the consequences of the LR1 restriction. Renaming a parameter does not alter these shared foundations.

The fixed rate ratio gives LR1 a consistent tail mechanism: the gamma survival component dominates at every positive parameter. Near the origin, however, the exponential component has positive density while the gamma density vanishes to third order. We prove that the combination forces an initial hazard decrease, followed eventually by an increase toward $\lambda/2$. Thus the hazard is nonmonotone throughout the parameter space. This endpoint statement does not assert that there is exactly one minimum. A separate density analysis is sharper: after expressing time in a dimensionless coordinate, a single explicit threshold in $\lambda$ determines whether the density has an interior local peak.

Inference has an equally specific structure. The density at zero is strictly increasing in $\lambda$, establishing identifiability. The mean decreases continuously from infinity to zero, so matching it to a positive sample mean yields exactly one moment estimate. Likelihood maximization remains a numerical problem and need not have a unique maximizer. We develop its score and observed information while accounting for unobserved component allocations. The expectation--maximization construction follows \citet{dempster1977}; the spacing alternative follows \citet{cheng1983}. Regular likelihood arguments give asymptotic inference on compact intervals, and log-scale confidence limits respect positivity.

The contribution is a model-specific account of distributional shape, reliability, and estimation under a fixed component-rate contrast. Explicit expressions cover transforms, moments, truncated quantities, quantiles, order statistics, residual-life summaries, and stress--strength reliability. The numerical study evaluates likelihood and moment estimation using eight settings from the original combined simulation experiment. No observed lifetime dataset has been analyzed, and the results do not establish superiority over other lifetime distributions. A four-stage interpretation of the gamma component remains a possible substantive explanation, requiring evidence in any proposed application.

Gamma notation throughout is shape and rate. All quantities refer to a declared working time unit: even though both component rates are proportional to $\lambda$, the parameter-dependent weight prevents the one-parameter family from being closed under arbitrary changes of that unit. The discussion describes a separate scale extension and its limits. ``LR1'' remains a working distribution label and does not denote a likelihood ratio.

Section~\ref{sec:construction} specifies the distribution and proves its basic validity and identifiability. Section~\ref{sec:shape} develops the density threshold and hazard results. Section~\ref{sec:moments} derives moments and associated distributional summaries, and Section~\ref{sec:reliability} treats reliability measures. Estimation and inference appear in Section~\ref{sec:inference}, followed by the numerical study in Section~\ref{sec:numerical} and discussion in Section~\ref{sec:discussion}. The appendix documents exact sampling and computational reproduction.

\section{The fixed-rate-ratio construction}\label{sec:construction}
Let $x>0$ denote a lifetime in the chosen working unit and let $\lambda>0$ be the LR1 parameter. Let $\Exp(a)$ have rate $a$, and let $\Gam(4,b)$ have shape four and rate $b$. Their densities are
\begin{equation}
 e_a(x)=a e^{-ax},\qquad g_b(x)=\frac{b^4x^3}{6}e^{-bx}.
\end{equation}
Define $w=w(\lambda)=\lambda^2/(1+\lambda^2)$ and let $Z$ be Bernoulli with success probability $w$. Conditional on $Z$, take
\begin{equation}\label{eq:latent}
 X\mid Z=1\sim\Exp(a(\lambda)),\qquad X\mid Z=0\sim\Gam(4,b(\lambda)).
\end{equation}
Thus the common density is
\begin{equation}\label{eq:generalpdf}
 f(x;\lambda)=w a e^{-ax}+(1-w)\frac{b^4x^3}{6}e^{-bx}.
\end{equation}
Hereafter, the auxiliary symbols always mean $a=2\lambda$ and $b=\lambda/2$; they are not independently estimated parameters. The LR1 density is explicitly
\begin{equation}\label{eq:specificpdf}
 f(x;\lambda)=\frac{2\lambda^3e^{-2\lambda x}+\lambda^4x^3e^{-\lambda x/2}/96}{1+\lambda^2},\qquad x>0.
\end{equation}

\begin{proposition}[Validity and survival function]\label{prop:valid}
LR1 is a probability distribution on $(0,\infty)$. Define
\begin{equation}
 P_3(y)=1+y+\frac{y^2}{2}+\frac{y^3}{6}.
\end{equation}
The survival function and distribution function are
\begin{equation}\label{eq:survival}
 S(x;\lambda)=w e^{-ax}+(1-w)e^{-bx}P_3(bx),\qquad F(x;\lambda)=1-S(x;\lambda).
\end{equation}
For $x\leq0$, $F(x;\lambda)=0$ and $S(x;\lambda)=1$.
\end{proposition}
\begin{proof}
Normalization follows immediately by integrating the mixture: each component contributes its allocation probability, and the two probabilities sum to one. Repeated integration by parts gives the gamma survival function $e^{-bx}P_3(bx)$, and conditioning on $Z$ gives \eqref{eq:survival}.
\end{proof}

For this model, the explicit survival function is
\begin{equation}\label{eq:specificsurvival}
 S(x;\lambda)=\frac{\lambda^2 e^{-2\lambda x}+e^{-\lambda x/2}P_3(\lambda x/2)}{1+\lambda^2}.
\end{equation}
The cumulative hazard is $H(x;\lambda)=-\log S(x;\lambda)$.

\begin{proposition}[Identifiability within a specification]\label{prop:ident}
For LR1, distinct positive parameters determine distinct distributions.
\end{proposition}
\begin{proof}
The continuous extension of the density to zero is
\begin{equation}\label{eq:origin}
 f(0+;\lambda)=wa=\frac{2\lambda^3}{1+\lambda^2},\qquad
 \frac{\dd}{\dd\lambda}f(0+;\lambda)=\frac{2\lambda^2(3+\lambda^2)}{(1+\lambda^2)^2}>0.
\end{equation}
Equal distributions have equal continuous densities on the positive support and therefore equal limits at zero. Strict monotonicity identifies $\lambda$.
\end{proof}

\section{Shape consequences of the LR1 restriction}\label{sec:shape}
\subsection{An explicit threshold for the density}
For LR1 the boundary behavior alone does not determine the density shape. The density starts with a decrease, but a subsequent increase and decrease can create a second, interior peak. The following criterion separates these possibilities.
\begin{theorem}[Density-shape threshold]\label{thm:mode}
Fix $0<w<1$ and $a,b>0$ in \eqref{eq:generalpdf}. For $0<x<3/b$, define
\begin{equation}\label{eq:Rshape}
 R(x)=\frac{(1-w)b^4}{6wa^2}x^2(3-bx)e^{(a-b)x}.
\end{equation}
This function has a unique maximizer $x_\star\in(0,3/b)$, characterized by
\begin{equation}\label{eq:xstar}
 \frac2{x_\star}-\frac b{3-bx_\star}+(a-b)=0.
\end{equation}
If $R(x_\star)<1$, the density is strictly decreasing. If $R(x_\star)=1$, it remains strictly decreasing but has one horizontal tangent. If $R(x_\star)>1$, the density has exactly two stationary points: a local minimum followed by a local maximum. No additional stationary points occur.
\end{theorem}
\begin{proof}
Differentiation yields
\begin{equation}\label{eq:pdfderivative}
 f'(x)=-wa^2e^{-ax}+(1-w)\frac{b^4}{6}x^2(3-bx)e^{-bx}.
\end{equation}
For $x\geq3/b$, both terms make $f'(x)<0$. On $(0,3/b)$, its sign is the sign of $R(x)-1$. Moreover,
\[
 (\log R)'(x)=\frac2x-\frac b{3-bx}+a-b,\qquad
 (\log R)''(x)=-\frac2{x^2}-\frac{b^2}{(3-bx)^2}<0.
\]
The first derivative decreases continuously from $+\infty$ to $-\infty$. Hence $R$ increases to a unique maximum and then decreases, with zero limits at both interval endpoints. Comparing that maximum with one proves all three cases and the signs on either side of each stationary point.
\end{proof}

For LR1, putting $y=bx=\lambda x/2$ reduces the threshold function to
\begin{equation}
 R(x)=\frac{y^2(3-y)e^{3y}}{96\lambda^2},\qquad 0<y<3.
\end{equation}
Equation~\eqref{eq:xstar} gives $y_\star=1+\sqrt3$ and $x_\star=2(1+\sqrt3)/\lambda$. Define
\begin{equation}
 \lambda_c=\left\{\frac{(1+\sqrt3)^2(2-\sqrt3)e^{3(1+\sqrt3)}}{96}\right\}^{1/2}.
\end{equation}
Then the density has a local minimum and an interior local maximum for $0<\lambda<\lambda_c$, is strictly decreasing with one horizontal tangent at $\lambda=\lambda_c$, and is strictly decreasing without stationary points when $\lambda>\lambda_c$. This classifies density stationary points, rather than hazard turning points.

\subsection{Initial decline and eventual recovery of the hazard}
The hazard is
\begin{equation}\label{eq:hazard}
 h(x;\lambda)=\frac{wae^{-ax}+(1-w)b^4x^3e^{-bx}/6}
 {we^{-ax}+(1-w)e^{-bx}P_3(bx)}.
\end{equation}
It is also a survival-weighted average of the component hazards:
\begin{equation}\label{eq:hazardmix}
 h(x)=\pi(x)a+\{1-\pi(x)\}h_G(x),\qquad
 \pi(x)=\frac{we^{-ax}}{S(x)},\quad
 h_G(x)=\frac{b^4x^3}{6P_3(bx)}.
\end{equation}
Here the allocation probabilities are conditional on surviving to time $x$. They therefore change with age, even though the original mixture probabilities are fixed for a given parameter. The increasing gamma-component hazard consequently does not imply an increasing population hazard.

\begin{theorem}[Initial and limiting hazard]\label{thm:hazard}
For LR1 and every $\lambda>0$,
\begin{equation}\label{eq:hzero}
 h(0+)=wa,\qquad h'(0+)=-w(1-w)a^2<0.
\end{equation}
Furthermore,
\begin{equation}\label{eq:htail}
 \lim_{x\to\infty}h(x)=\min(a,b).
\end{equation}
If $b\leq a$, then
\begin{equation}\label{eq:gammatail}
 h(x)=b-\frac3x+O(x^{-2}),\qquad
 h'(x)=\frac3{x^2}+O(x^{-3}),
\end{equation}
so the hazard eventually increases and cannot be globally monotone. If $a<b$, then
\begin{equation}\label{eq:exptail}
 h(x)-a\sim\frac{1-w}{w}\frac{b^3(b-a)}6x^3e^{-(b-a)x},
\end{equation}
and the hazard is eventually decreasing toward $a$ from above.
\end{theorem}
\begin{proof}
Since $S(0)=1$, $f(0)=wa$, $S'(0)=-wa$, and $f'(0)=-wa^2$, differentiating $h=f/S$ gives \eqref{eq:hzero}. If $b<a$, the gamma survival term dominates exponentially. If $b=a$, its cubic polynomial factor still dominates the exponential component. In either case, expansion of the explicit polynomial ratio in \eqref{eq:hazard} gives \eqref{eq:gammatail}; the exponentially small corrections when $b<a$ and the lower polynomial terms when $b=a$ can be differentiated to give the displayed derivative. If $a<b$, use
\[
 h-a=\frac{(1-w)e^{-bx}\{b^4x^3/6-aP_3(bx)\}}{S(x)}.
\]
The numerator's leading polynomial coefficient is $b^3(b-a)/6$, and $S(x)\sim we^{-ax}$. This proves \eqref{eq:exptail}; differentiation of this explicit rational-exponential expression shows eventual decrease. The common limit follows in all cases.
\end{proof}

\begin{corollary}\label{cor:nonmono}
For LR1, $b=\lambda/2<a=2\lambda$ for every $\lambda>0$. Consequently,
\begin{equation}
 h(0+)=\frac{2\lambda^3}{1+\lambda^2},\quad
 h'(0+)=-\frac{4\lambda^4}{(1+\lambda^2)^2},\quad
 h(\infty)=\frac\lambda2.
\end{equation}
The hazard is initially decreasing and eventually increasing for every positive parameter; hence it is never globally monotone. The tail expansion is $h(x)=\lambda/2-3/x+O(x^{-2})$.

\end{corollary}
\begin{proof}
Substitute the LR1 rates into Theorem~\ref{thm:hazard}.
\end{proof}
The corollary does not assert exactly one hazard minimum. Counting all hazard turning points requires additional analysis; the density criterion in Theorem~\ref{thm:mode} is a different result.

\begin{figure}[htbp]\centering
\includegraphics[width=\textwidth]{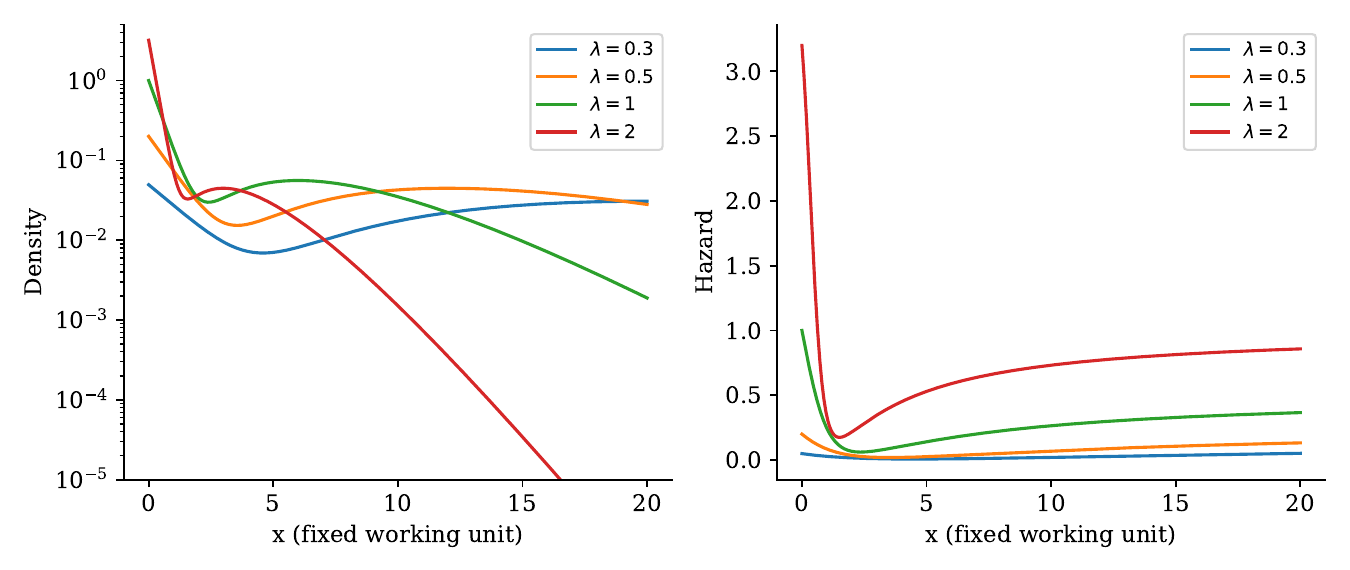}
\caption{LR1 density (left, logarithmic vertical scale) and hazard (right) curves for $\lambda=0.3,0.5,1,2$. The finite plotting range illustrates behavior but does not replace the tail results of Theorem~\ref{thm:hazard}.}\label{fig:shapes}
\end{figure}

\section{Distributional summaries and transforms}\label{sec:moments}
\begin{proposition}[Raw moments and transforms]\label{prop:moments}
For every real $r>-1$,
\begin{equation}\label{eq:moments}
 \mu'_r=\E(X^r)=w\frac{\Gamma(r+1)}{a^r}
 +(1-w)\frac{\Gamma(r+4)}{6b^r}.
\end{equation}
The moment-generating function is
\begin{equation}\label{eq:mgf}
 M(t)=w\frac{a}{a-t}+(1-w)\left(\frac{b}{b-t}\right)^4,
 \qquad t<\min(a,b).
\end{equation}
The Laplace transform is $M(-s)$ for $s\geq0$, and the characteristic function is obtained by replacing $t$ by $it$ in \eqref{eq:mgf}.
\end{proposition}
\begin{proof}
Integrate $x^r$ against the two component densities and substitute $y=ax$ and $y=bx$. The exponential integral requires $r>-1$, which also ensures finiteness of the gamma integral. Integration of $e^{tx}$ gives \eqref{eq:mgf}; because both weights are positive, finiteness requires $t$ to be smaller than both rates.
\end{proof}

Writing $\mu=\mu'_1$, the first four raw moments are
\begin{align}
 \mu&=\frac{w}{a}+\frac{4(1-w)}{b},\label{eq:mean}\\
 \mu'_2&=\frac{2w}{a^2}+\frac{20(1-w)}{b^2},\\
 \mu'_3&=\frac{6w}{a^3}+\frac{120(1-w)}{b^3},\\
 \mu'_4&=\frac{24w}{a^4}+\frac{840(1-w)}{b^4}.
\end{align}
In particular, the law of total variance gives the interpretable decomposition
\begin{equation}\label{eq:variance}
 \sigma^2=\frac{w}{a^2}+\frac{4(1-w)}{b^2}
 +w(1-w)\left(\frac1a-\frac4b\right)^2.
\end{equation}
Within-component variation gives the first two terms; the final term accounts for separation of the component means. Standardized summaries then follow as
\begin{align}
 \mathrm{CV}&=\sigma/\mu,\\
 \gamma_1&=(\mu'_3-3\mu\mu'_2+2\mu^3)/\sigma^3,\\
 \gamma_2&=(\mu'_4-4\mu\mu'_3+6\mu^2\mu'_2-3\mu^4)/\sigma^4-3.
\end{align}
These expressions determine the dispersion and standardized shape measures of LR1.

For LR1, substitution of the rates yields the fully parameterized raw-moment formula
\begin{equation}\label{eq:specificmoments}
 \mu'_r=\frac{2^{-r}\lambda^{2-r}\Gamma(r+1)+2^r\lambda^{-r}\Gamma(r+4)/6}{1+\lambda^2},\qquad r>-1.
\end{equation}
In particular,
\begin{equation}\label{eq:specificmeans}
 \mu(\lambda)=\frac{\lambda^2+16}{2\lambda(1+\lambda^2)},\qquad \mu'_2=\frac{1/2+80/\lambda^2}{1+\lambda^2}.
\end{equation}
The variance can equivalently be computed as $\sigma^2=\mu'_2-\mu^2$.
\begin{proposition}[Monotonicity of the mean]\label{prop:meanmono}
The LR1 mean is continuous and strictly decreasing from infinity to zero on $(0,\infty)$.
\end{proposition}
\begin{proof}
Differentiating gives
\begin{equation}
 \frac{\dd\mu(\lambda)}{\dd\lambda}=-\frac{\lambda^4+47\lambda^2+16}{2\lambda^2(1+\lambda^2)^2}.
\end{equation}
The displayed derivative is strictly negative for every positive parameter. The endpoint limits follow directly from \eqref{eq:specificmeans}.
\end{proof}

\begin{figure}[htbp]\centering
\includegraphics[width=.95\textwidth]{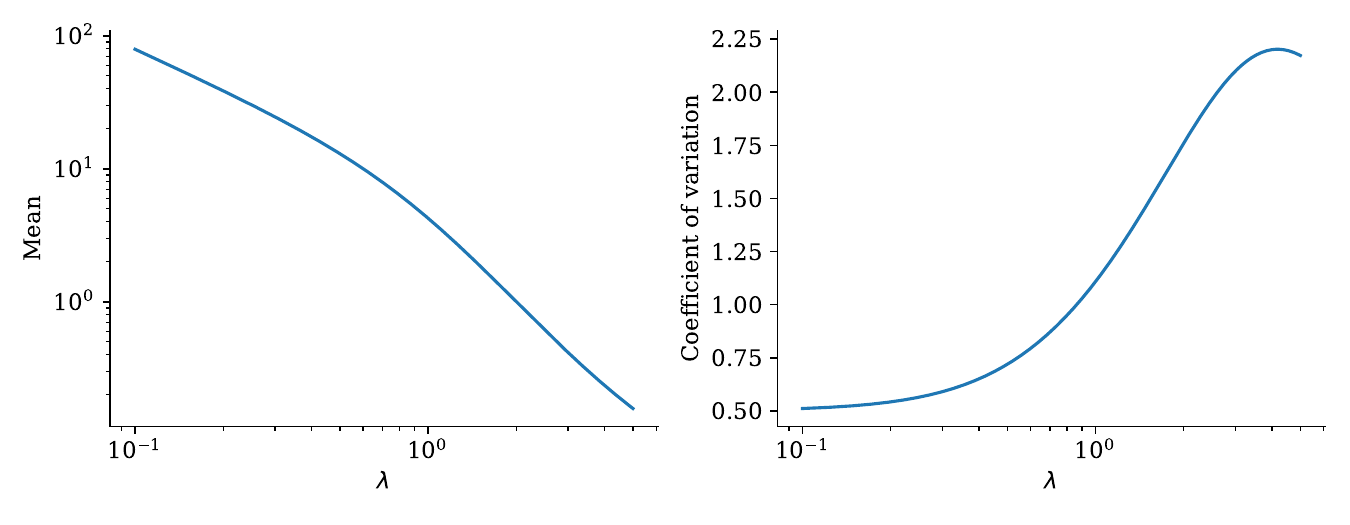}
\caption{Mean and coefficient of variation for LR1. These are direct evaluations of the analytic moment formulas in a fixed working time unit.}\label{fig:moments}
\end{figure}

\subsection{Partial moments and ordered lifetimes}
Let $\gamma(\nu,z)=\int_0^z u^{\nu-1}e^{-u}\dd u$ and $\Gamma(\nu,z)=\int_z^\infty u^{\nu-1}e^{-u}\dd u$ denote the lower and upper incomplete gamma functions. For $r>-1$,
\begin{align}
 T_r(x)&=\E[X^r\mathbf1\{X\leq x\}]
 =w a^{-r}\gamma(r+1,ax)+(1-w)\frac{b^{-r}}6\gamma(r+4,bx),\label{eq:lower}\\
 U_r(x)&=\E[X^r\mathbf1\{X>x\}]
 =w a^{-r}\Gamma(r+1,ax)+(1-w)\frac{b^{-r}}6\Gamma(r+4,bx).\label{eq:upper}
\end{align}
Because $f(x;\lambda)>0$, every $p\in(0,1)$ has a unique quantile $Q(p;\lambda)$. It can be found by solving $S(Q;\lambda)=1-p$ using a bracketed root finder. There is no need to assume an elementary inverse CDF. For a numerically stable bracket, start at zero, choose any positive upper value, and double it until the survival probability is below $1-p$.

For independent observations of size $n$, the $k$th order statistic has density
\begin{equation}
 f_{k:n}(x)=\frac{n!}{(k-1)!(n-k)!}
 F(x)^{k-1}S(x)^{n-k}f(x).
\end{equation}
For example, the minimum has survival $S(x)^n$ and hazard $n h(x)$, corresponding to an independent series system. These identities require no additional distributional assumptions beyond independence.

The Lorenz curve, when this measure is appropriate for a positive quantity, is $L(p)=T_1(Q(p))/\mu$. The mean absolute deviation about the mean is
\begin{equation}
 \E|X-\mu|=2\{\mu F(\mu)-T_1(\mu)\}.
\end{equation}
Both expressions are consequences of the truncated first moment, rather than separate special-purpose derivations.

\section{Remaining lifetime and stress--strength calculations}\label{sec:reliability}
\subsection{Mean residual life and mean inactivity time}
The mean residual life at age $x$ is $m(x)=\E(X-x\mid X>x)$. Integrating \eqref{eq:survival} term by term gives
\begin{equation}\label{eq:mrl}
 m(x)=\frac{w e^{-ax}/a+(1-w)e^{-bx}B_3(bx)/b}{S(x)},
 \quad B_3(y)=4+3y+y^2+\frac{y^3}{6}.
\end{equation}
Equivalently, $m(x)=U_1(x)/S(x)-x$. The mean inactivity time is
\begin{equation}
 \widetilde m(x)=\E(x-X\mid X\leq x)=x-\frac{T_1(x)}{F(x)},\qquad x>0.
\end{equation}
These standard lifetime summaries complement the hazard; see \citet{cox1984} for the survival-analysis setting. At zero, $m(0)=\mu$, and its limiting value is $1/\min(a,b)$. The latter follows directly from \eqref{eq:mrl} and the dominant survival component.

\subsection{Stress--strength reliability}
Let $X$ denote strength and $Y$ stress, with independent mixture distributions. Both follow LR1 and may have different parameter values. Write $(w_X,a_X,b_X)$ and $(w_Y,a_Y,b_Y)$ for their component specifications. Define reliability as $\mathcal R=\Prob(Y<X)$; with this convention,
\begin{equation}\label{eq:reliability}
 \mathcal R=w_Xw_Y R_{EE}+w_X(1-w_Y)R_{EG}
 +(1-w_X)w_YR_{GE}+(1-w_X)(1-w_Y)R_{GG},
\end{equation}
where the first subscript refers to the strength component and
\begin{align}
 R_{EE}&=\frac{a_Y}{a_X+a_Y},\\
 R_{EG}&=\left(\frac{b_Y}{b_Y+a_X}\right)^4,\\
 R_{GE}&=1-\left(\frac{b_X}{b_X+a_Y}\right)^4,\\
 R_{GG}&=1-\sum_{k=0}^3\frac{\Gamma(k+4)}{6k!}
 \frac{b_X^4b_Y^k}{(b_X+b_Y)^{k+4}}.\label{eq:RGG}
\end{align}
To verify these expressions, condition on both latent components. For $R_{EG}$, integrate $S_{\Exp(a_X)}(y)$ against the gamma stress density, giving its Laplace transform. For $R_{GE}$, integrate the exponential stress CDF against the gamma strength density. For $R_{GG}$, integrate $S_Y(x)$ against the strength density and subtract from one. Equal distributions yield $\mathcal R=1/2$, as required by independence and continuity.

\section{Estimating the LR1 parameter}\label{sec:inference}
\subsection{Inverting the mean}
Suppose $X_1,\ldots,X_n$ are independent and identically distributed from LR1. The moment estimator $\widehat\lambda_{\rm MM}$ solves
\begin{equation}\label{eq:momest}
 \mu(\widehat\lambda_{\rm MM})=\overline X.
\end{equation}
Proposition~\ref{prop:meanmono} proves existence and uniqueness for every strictly positive sample mean. Multiplication by the positive denominator gives the model-specific equation
\begin{equation}
 2\overline X\lambda^3-\lambda^2+2\overline X\lambda-16=0.
\end{equation}
Despite the polynomial representation, a monotone root solver on $\eta=\log\lambda$ is more convenient than explicit radical formulas. The law of large numbers and continuity of the inverse mean prove consistency. Since all moments are finite, the central limit theorem and delta method give
\begin{equation}\label{eq:momasymptotic}
 \sqrt n(\widehat\lambda_{\rm MM}-\lambda_0)
 \xrightarrow{d}N\left(0,\frac{\sigma^2(\lambda_0)}{\{\mu'(\lambda_0)\}^2}\right).
\end{equation}
The variance in \eqref{eq:momasymptotic} is readily estimated by plugging the moment estimate into the analytic expressions.

\subsection{Likelihood construction and curvature}
Define the weighted component densities
\begin{equation}
 A_i=wae^{-aX_i},\qquad B_i=(1-w)b^4X_i^3e^{-bX_i}/6.
\end{equation}
The log likelihood and posterior component probability are
\begin{equation}\label{eq:likelihood}
 \ell(\lambda)=\sum_{i=1}^n\log(A_i+B_i),\qquad
 r_i(\lambda)=\frac{A_i}{A_i+B_i}.
\end{equation}
Let primes on rates denote derivatives with respect to $\lambda$. The weighted-component scores are
\begin{align}
 s_{Ei}&=\frac2\lambda-\frac{2\lambda}{1+\lambda^2}+\frac{a'}a-a'X_i,\label{eq:sE}\\
 s_{Gi}&=-\frac{2\lambda}{1+\lambda^2}+4\frac{b'}b-b'X_i.\label{eq:sG}
\end{align}
Then
\begin{equation}\label{eq:score}
 U(\lambda)=\ell'(\lambda)=\sum_{i=1}^n\{r_i s_{Ei}+(1-r_i)s_{Gi}\}.
\end{equation}
Writing $c(\lambda)=-2(1-\lambda^2)/(1+\lambda^2)^2$, the score derivatives are
\begin{align}
 s'_{Ei}&=-\frac2{\lambda^2}+c(\lambda)+\frac{a''}a-\left(\frac{a'}a\right)^2-a''X_i,\\
 s'_{Gi}&=c(\lambda)+4\left\{\frac{b''}b-\left(\frac{b'}b\right)^2\right\}-b''X_i.
\end{align}
Differentiating the posterior allocation gives $r_i'=r_i(1-r_i)(s_{Ei}-s_{Gi})$. Hence the negative likelihood curvature is
\begin{equation}\label{eq:info}
 J_n(\lambda)=-\sum_{i=1}^n\left\{r_i s'_{Ei}+(1-r_i)s'_{Gi}
 +r_i(1-r_i)(s_{Ei}-s_{Gi})^2\right\}.
\end{equation}
The final variance term is essential; omitting it treats unobserved component assignments as known. For LR1, the rate derivatives are
\begin{equation}
 a'=2,\quad a''=0,\quad b'=1/2,\quad b''=0.
\end{equation}
The component scores therefore simplify to
\begin{equation}
  s_{Ei}=\frac3\lambda-\frac{2\lambda}{1+\lambda^2}-2X_i,\qquad s_{Gi}=\frac4\lambda-\frac{2\lambda}{1+\lambda^2}-\frac{X_i}{2}.
\end{equation}
These explicit expressions and \eqref{eq:info} are used in the numerical implementation.

For a sample containing only positive lifetimes, the log likelihood tends to $-\infty$ at both ends of the positive parameter space. At the lower endpoint, each weighted component density tends to zero with a positive power of $\lambda$. At the upper endpoint, the exponential decay dominates the polynomial prefactors. Continuity therefore ensures at least one interior likelihood maximizer. Uniqueness is not asserted.

For numerical estimation, evaluate the likelihood using the log-sum-exp identity
\[
 \log(A+B)=m+\log\{e^{\log A-m}+e^{\log B-m}\},\quad
 m=\max(\log A,\log B).
\]
Optimize in $\eta=\log\lambda$, search a grid for candidate peaks, refine each candidate, and retain the largest likelihood. Inspect boundaries and expand the range if they are competitive. A numerical grid is a practical safeguard, not a proof of finding every finite-sample stationary point.

\subsection{Latent allocations and constrained EM}
The latent representation \eqref{eq:latent} yields an expectation--maximization (EM) implementation in the sense of \citet{dempster1977}. Given $\lambda^{(q)}$, calculate $r_i^{(q)}=r_i(\lambda^{(q)})$ in the E step. Up to terms independent of $\lambda$, the M-step objective is
\begin{align}\label{eq:EM}
 Q(\lambda\mid\lambda^{(q)})={}&\sum_i r_i^{(q)}
 \{\log w(\lambda)+\log a(\lambda)-a(\lambda)X_i\}\notag\\
 &+\sum_i(1-r_i^{(q)})
 \{\log(1-w(\lambda))+4\log b(\lambda)-b(\lambda)X_i\}.
\end{align}
Maximize this scalar function subject to the LR1 rate specification. Updating the weight and two rates independently would fit a different, unconstrained model. Exact maximization of $Q$ makes the observed likelihood nondecreasing, but EM can converge to a local stationary point. Multiple starts and comparison with the direct likelihood search remain useful. The simulation below uses direct likelihood optimization, not EM.

\subsection{Asymptotic uncertainty and confidence limits}
\begin{theorem}[Regular likelihood inference]\label{thm:asymptotic}
Let $\lambda_0$ be an interior point of a fixed compact parameter interval $K\subset(0,\infty)$. For independent observations from LR1, a global maximizer over $K$ is consistent. Moreover,
\begin{equation}\label{eq:MLEasymptotic}
 \sqrt n(\widehat\lambda_{\rm ML}-\lambda_0)
 \xrightarrow{d}N(0,I(\lambda_0)^{-1}),\qquad
 I(\lambda)=\int_0^\infty\{\partial_\lambda\log f(x;\lambda)\}^2f(x;\lambda)\dd x,
\end{equation}
where $0<I(\lambda_0)<\infty$.
\end{theorem}
\begin{proof}
On $K$, weights and rates are bounded away from zero and infinity. The score and its first two parameter derivatives are bounded by polynomials in $x$, since posterior component probabilities lie between zero and one and rate derivatives are bounded on $K$. Component tails admit a uniform exponential envelope. The log density admits an integrable envelope of the form $C(1+x+|\log x|)$; in fact the positive exponential component controls the behavior near zero. These bounds give uniform laws of large numbers and justify differentiating under the integral. Identifiability (Proposition~\ref{prop:ident}) makes the expected log likelihood uniquely maximized at $\lambda_0$ by the Kullback--Leibler inequality. Compactness then proves consistency. Finiteness of $I$ follows from the score bound and finite moments. If $I=0$, $\partial_\lambda f(x;\lambda_0)$ vanishes almost everywhere and, by continuity, everywhere on the positive support. Its continuous limit at zero would then vanish, contradicting the strictly positive derivative of \eqref{eq:origin}. Thus $I>0$. A score Taylor expansion at the consistent maximizer, the central limit theorem, and convergence of the normalized Hessian give \eqref{eq:MLEasymptotic}; these are the standard regular likelihood arguments described by \citet{vandervaart1998}.
\end{proof}

These asymptotic conclusions apply to an interior true parameter on the stated compact interval. They give neither a boundary-uniform approximation nor a certificate that a numerical local solution is globally optimal. At a fitted solution with $J_n(\widehat\lambda)>0$, an approximate $100(1-\alpha)\%$ interval that respects positivity is
\begin{equation}\label{eq:wald}
 \left[\exp\left\{\log\widehat\lambda-
 \frac{z_{1-\alpha/2}}{\widehat\lambda\sqrt{J_n(\widehat\lambda)}}\right\},
 \exp\left\{\log\widehat\lambda+
 \frac{z_{1-\alpha/2}}{\widehat\lambda\sqrt{J_n(\widehat\lambda)}}\right\}\right].
\end{equation}
Here $J_n$ is total information, so no additional factor of $n$ is inserted. A likelihood-based alternative is the set
\begin{equation}
 \{\lambda>0:2[\ell(\widehat\lambda)-\ell(\lambda)]\leq\chi^2_{1,1-\alpha}\}.
\end{equation}
This cutoff concerns inference within LR1 under the stated regularity assumptions. Parametric bootstrap intervals can also be obtained by generating from the fitted model and refitting with the same search procedure.

For a differentiable functional $g(\lambda)$, including a survival probability or the mean, the delta-method variance is $\{g'(\lambda_0)\}^2/[nI(\lambda_0)]$. A quantile derivative follows implicitly from $F(Q(p;\lambda);\lambda)=p$:
\begin{equation}
 \partial_\lambda Q(p;\lambda)=-\frac{\partial_\lambda F(Q(p;\lambda);\lambda)}{f(Q(p;\lambda);\lambda)}.
\end{equation}

\subsection{Spacing criteria and censored observations}
For distinct, uncensored observations ordered as $X_{(1)}<\cdots<X_{(n)}$, let $F(X_{(0)})=0$ and $F(X_{(n+1)})=1$. The maximum product of spacings estimator maximizes
\begin{equation}\label{eq:spacing}
 \sum_{i=1}^{n+1}\log\{F(X_{(i)};\lambda)-F(X_{(i-1)};\lambda)\}.
\end{equation}
The general approach is due to \citet{cheng1983}. Stable differences of survival probabilities are useful in the upper tail. Ties caused by rounding produce zero spacings and require a grouped-data or tie-adjusted method; \eqref{eq:spacing} should not be applied unchanged. No numerical comparison of spacing estimation is claimed in this paper.

For independent right censoring, observe $T_i=\min(X_i,C_i)$ and $\delta_i=\mathbf1\{X_i\leq C_i\}$. Assuming the censoring mechanism is independent of $X_i$ and its distribution has no shared unknown parameter, the part of the log likelihood depending on $\lambda$ is
\begin{equation}\label{eq:censored}
 \ell_c(\lambda)=\sum_i\{\delta_i\log f(T_i;\lambda)+(1-\delta_i)\log S(T_i;\lambda)\}.
\end{equation}
The gamma-component posterior probability for a censored observation uses its survival contribution, not its density contribution. In particular, the exponential posterior probability is $we^{-aT_i}/S(T_i)$ for $\delta_i=0$. The complete-data formulation must accordingly include survival terms for censored records. Equation~\eqref{eq:EM}, written for complete lifetimes, is not the censored-data M-step. The uncensored moment estimator is likewise inappropriate when applied naively to $T_i$. Censored estimation is an extension described by its likelihood here and is not evaluated in the simulation.

\section{Finite-sample assessment}\label{sec:numerical}
\subsection{Simulation design and numerical safeguards}
We report a Monte Carlo experiment for LR1 with $\lambda\in\{0.3,0.5,1,2\}$ and $n\in\{50,200\}$. Each of the eight settings uses 500 independently generated samples, for 4,000 datasets in total. These results are the LR1 subset of the combined experiment from which this standalone analysis was developed, not a new independent experiment. The original random stream uses NumPy's \texttt{default\_rng} with seed 20261009. The accompanying script preserves the stream position used for this model before generating its samples. Indicators use $w(\lambda)$, and lifetimes are drawn from the appropriate exponential or gamma component. No observed lifetime dataset was used.

For each sample, we computed the moment estimator by a bracketed root search on the log-parameter scale, and the likelihood estimator by evaluating 121 equally spaced log-parameter values on $[-9,9]$, refining all grid-local maxima with bounded scalar optimization. The retained candidate had the largest likelihood. Endpoint comparisons checked that a boundary was not preferred. This is a numerical approximation to global maximization with an explicitly documented search range.

We report empirical bias and root mean squared error (RMSE),
\begin{equation}
 \widehat{\mathrm{Bias}}=\frac1B\sum_{r=1}^B(\widehat\lambda_r-\lambda),\qquad
 \widehat{\mathrm{RMSE}}=\left\{\frac1B\sum_{r=1}^B(\widehat\lambda_r-\lambda)^2\right\}^{1/2},
 \quad B=500.
\end{equation}
Coverage refers to the nominal 95\% log-scale Wald interval in \eqref{eq:wald}. The machine-readable output also records optimizer failures, nonpositive observed-information counts, and Monte Carlo standard errors of coverage. An unavailable interval is recorded separately rather than silently treated as a valid result. With 500 replications, the standard error of a coverage estimate near 0.95 is about 0.010. Modest coverage differences should therefore be interpreted cautiously.

Before simulation, numerical quadrature verified normalization and the first two raw moments for LR1 at the four parameter values. Finite differences checked the analytic score and Hessian at three positive $x$ values per setting; differentiating the survival function numerically also checked the density. These checks support the implementation but do not substitute for the proofs.

\begin{table}[htbp]\centering\small
\caption{Monte Carlo bias, RMSE, and 95\% log-scale Wald coverage, based on 500 replications per setting. ML denotes maximum likelihood and MM the moment estimator.}\label{tab:simulation}
\begin{tabular}{crrrrrr}
\toprule
$\lambda$ & $n$ & Bias (ML) & RMSE (ML) & Bias (MM) & RMSE (MM) & Coverage \\
\midrule
0.3 & 50 & 0.0010 & 0.0222 & 0.0009 & 0.0226 & 0.940 \\
0.3 & 200 & 0.0017 & 0.0110 & 0.0017 & 0.0111 & 0.942 \\
0.5 & 50 & 0.0032 & 0.0375 & 0.0032 & 0.0383 & 0.948 \\
0.5 & 200 & -0.0002 & 0.0184 & 0.0003 & 0.0194 & 0.938 \\
1 & 50 & 0.0073 & 0.0790 & 0.0072 & 0.0861 & 0.954 \\
1 & 200 & -0.0007 & 0.0392 & 0.0014 & 0.0410 & 0.948 \\
2 & 50 & 0.0424 & 0.2244 & 0.0671 & 0.2710 & 0.948 \\
2 & 200 & 0.0154 & 0.0999 & 0.0154 & 0.1138 & 0.940 \\
\bottomrule
\end{tabular}

\end{table}

\subsection{Results}
Table~\ref{tab:simulation} shows that increasing $n$ from 50 to 200 reduces RMSE at every reported parameter value for both estimators. Maximum likelihood has lower empirical RMSE than moment estimation in 8 of the eight settings. At $\lambda=2$ and $n=50$, ML and MM RMSE are 0.2244 and 0.2710; at $n=200$ they are 0.0999 and 0.1138. These are within-model estimator comparisons, rather than evidence of superiority over competing lifetime distributions.

Observed 95\% interval coverage ranges from 0.938 to 0.954. The lowest reported coverage occurs at $\lambda=0.5$ and $n=200$. With 500 replications per setting, Monte Carlo uncertainty is approximately one percentage point near nominal coverage. All 4,000 likelihood fits in this article returned successful optimization results, no boundary comparison required expansion, and all fitted observed-information values were positive. These finite-sample results support approximate inference within the reported design; they do not guarantee nominal coverage outside it.

\subsection{Interpretation and empirical scope}
The experiment isolates estimation error when LR1 is the data-generating distribution. It does not compare predictive performance against the Lindley, Weibull, lognormal, or unrestricted exponential--gamma mixture distributions, and it is not evidence of empirical superiority. A future application should specify the data provenance, working time unit, censoring and rounding mechanisms, and sample exclusions before fitting. Positive observations alone are insufficient to validate a lifetime interpretation.

For observed data, a reasonable comparison would fit a constant-hazard exponential model, the Lindley model, flexible two-parameter lifetime models, and an unrestricted exponential--gamma mixture with shape four. For a regular candidate model with $k$ fitted parameters, the Akaike criterion is $\mathrm{AIC}=-2\ell(\widehat\lambda)+2k$ \citep{akaike1974}. The fixed-unit LR1 model has $k=1$, so its criterion is $-2\ell(\widehat\lambda)+2$. If an additional time scale is estimated, it must count as an additional parameter. These are comparison principles, not fitted results.

Probability plots and held-out log likelihood can supplement information criteria. A goodness-of-fit statistic evaluated after estimating parameters does not generally have the same null distribution as when parameters are known. For example, a parametric-bootstrap CDF test should simulate from the fitted model, refit each bootstrap sample, and recompute the statistic. A real-data results table is deliberately omitted because no application dataset has been analyzed.

\section{Implications and limitations}\label{sec:discussion}
LR1 is a one-dimensional restriction of an exponential--gamma mixture with gamma shape four. The common representation yields explicit distributional and reliability quantities, while the constraints permit unique moment estimation and identifiable likelihood inference for the specified model. The density-turning-point threshold and the hazard results give more informative shape statements than visual inspection alone. In particular, an increasing component hazard is not enough to justify describing LR1 as an increasing-hazard distribution.

The fixed ratio of component rates leads to an explicit density-shape threshold in $\lambda$, while the hazard is nonmonotone throughout the parameter space. This is a useful distinction: density stationary points may disappear as the parameter changes even though the hazard retains an initial decrease and an eventual increase.

Several restrictions matter in applications. First, the component rates and the mixture weight cannot vary independently. This can make a fit parsimonious, but can also force an unrealistic relationship between the frequency and duration of the two mechanisms. Second, the fixed shape four is a modeling choice, not a finding established by observed data. Third, the current numerical study is limited to complete, independent, correctly specified samples and a finite parameter grid. It does not establish robustness to censoring, dependence, rounding, or misspecification.

The time-unit issue deserves particular attention. If $Y=cX$, then its component rates become $a/c$ and $b/c$, while the mixture weight remains $w(\lambda)$. Because $w$ is injective, representing $Y$ by the same one-parameter model requires the same $\lambda$; that would leave both rates unchanged. Thus LR1 is not closed under arbitrary time-unit changes. The linked weight and rates are interpreted numerically in that working unit; the parameter does not independently adjust the time scale.

A scale extension resolves this practical restriction. If $X$ denotes a main-model lifetime in the reference unit, set $Y=\tau X$ with $\tau>0$. Then
\begin{equation}
 f_Y(y;\lambda,\tau)=\tau^{-1}f_X(y/\tau;\lambda),\qquad
 S_Y(y;\lambda,\tau)=S_X(y/\tau;\lambda).
\end{equation}
This is a two-parameter model with rates $a(\lambda)/\tau$ and $b(\lambda)/\tau$. Under $Y^*=cY$, only $\tau$ changes to $c\tau$. Estimation and identifiability of this extension require their own analysis and are not covered by the one-parameter inference theorem.

Further work could evaluate this scale extension, relax the fixed gamma shape, and assess censoring and empirical performance with appropriately documented datasets. A full classification of hazard turning points would also strengthen the shape analysis. The present results establish mathematical and computational foundations for this particular restriction, while keeping the distinction between analytic properties and demonstrated application performance explicit.

\section*{Data and code availability}
All numerical results use generated data. The accompanying Python script \texttt{reproduce.py} regenerates both figures, the Monte Carlo table, the numerical-check report, and software-version metadata. The source bundle includes the CSV output and environment record. No confidential or proprietary data are involved.

\appendix
\section{Sampling and computational implementation}
For LR1 at a specified parameter, generate independent random variables $U\sim\operatorname{Uniform}(0,1)$, $E\sim\Exp(a)$, and $G\sim\Gam(4,b)$, and return
\begin{equation}
 X=E\mathbf1\{U<w\}+G\mathbf1\{U\geq w\}.
\end{equation}
This is exact mixture sampling and requires no inverse-CDF approximation. The gamma variable can alternatively be obtained as the sum of four independent exponentials of rate $b$. The provided implementation draws all three vectors for each replication so that the order of random-number consumption is fixed.

The Python dependencies are NumPy, SciPy, and Matplotlib. The actual versions used for the reported experiment are recorded in \texttt{environment.json}. Running \texttt{python3 anc/reproduce.py --replications 500} regenerates the outputs. The full likelihood and observed-information formulas are implemented rather than approximating information with a numerically unstable second difference. Numerical finite differences are used only for independent implementation checks.

\bibliographystyle{plainnat}
\bibliography{references}
\end{document}